\documentclass[a4paper,11pt]{article}
\usepackage{amsthm,amsmath,amssymb,amsfonts}
\usepackage[numbers, sort & compress]{natbib}
\usepackage{hyperref}
\usepackage{arydshln}

\makeatletter
\def\p@enumii{}
\makeatother

\newtheorem{thm}{Theorem}[section]
\newtheorem{lem}[thm]{Lemma}

\newtheorem{prop}[thm]{Proposition}
\newtheorem{cor}[thm]{Corollary}
\newtheorem{conj}{Conjecture}
\theoremstyle{definition}
\newtheorem{defn}[thm]{Definition}
\newtheorem{ex}[thm]{Example}
\theoremstyle{remark}
\newtheorem{rem}[thm]{Remark}
\numberwithin{equation}{section}

\newcommand{\z}{\mathbb{Z}}
\newcommand{\n}{\mathbb{N}}

\newcommand{\ifof}{if and only if }
\newcommand{\good}{LCED}

\renewcommand{\l}{\left}
\renewcommand{\r}{\right}
\newcommand{\chr}{\mathrm{char}\,}
\renewcommand{\sc}{\mathrm{sec}}
\renewcommand{\d}{\mathrm{d}}
\def\f{\frac}
\newcommand{\diag}{\mathrm{diag}}
\newcommand{\tr}{\mathrm{tr}}
\newcommand{\id}{\mathrm{Id}}
\newcommand{\rank}{\mathrm{rank}\,}
\newenvironment{psmat}
{\left(\begin{smallmatrix}}
	{\end{smallmatrix}\right)}
\newcommand{\leg}[2]{\l(\frac{#1}{#2}\r)}
\newcommand{\tleg}[2]{\big(\frac{#1}{#2}\big)}

\sloppypar

\begin{document}
\title{Linear Complementary Equi-Dual Codes}
\author{Ashkan Nikseresht$^1$, Shohreh Namazi$^2 $ and Marziyeh Beygi Khormaei$^3$   \\
\it\small Department of Mathematics, College of Sciences, Shiraz University, \\
\it\small 71457-44776, Shiraz, Iran\\
\it\small $^1$ E-mail: ashkan\_nikseresht@yahoo.com\\
\it\small $^2$ E-mail: namazi@shirazu.ac.ir\\
\it\small $^3$ E-mail: m64.beygi@gmail.com  }
\date{}
\maketitle
\begin{abstract}
We call a linear code $C$ with length $n$ over a field $F$, a linear complementary equi-dual code, when there exists
a linear code $D$ over $F$ such that $D$ is permutation equivalent to $C^\perp$ and $(C,D)$ is a linear complementary
pair of codes, that is, $C+ D=F^n$ and $C\cap D=0$. We first state a necessary condition on a code $C$ to be linear
complementary equi-dual. Then, we conjecture that this necessary condition is also sufficient and present several
statements which support this conjecture.
\end{abstract}
Keywords: LCP of codes, linear codes, code equivalence, permutation matrices\\
\indent 2020 Mathematical Subject Classification: 94B05, 94B60, 94A60.

\section{Introduction}
Linear complementary dual codes which were introduced by Massey in \cite{Massey} and their more recent generalizations,
linear complementary pairs of codes, have been extensively studied from the view points of algebraic and coding
properties and also applications in cryptography (see \cite{carlet crypto,Bhasin,ODSM,
Borello,CarletG,CarletM,GuneriM,GuneriO,Hu,Jin,LiuL,LiuX,LiuW,Massey,Sendrier} and the references therein). Recall that
a pair of linear codes $(C,D)$ with length $n$ over the field $F$, is called a \emph{linear complementary pair} of
codes (LCP of codes, for short), when $C\cap D=0$ and $C+ D=F^n$. In the case that $D=C^\perp$, the code $C$ is called
a linear complementary dual code (LCD, for short).

In \cite{Sendrier}, it is shown that LCD codes meet the asymptotic Gilbert-Varshamov bound and in \cite{Jin} several
classes of maximum distance separable LCD codes are introduced. In \cite{Borello,CarletG,GuneriM,GuneriO}, it is
studied when  for an LCP of codes such as $(C,D)$, both $C$ and $D$ are in certain classes of codes, such as
($n$-dimensional) cyclic or quasi-cyclic codes. In \cite{LiuX,LiuW} and \cite{GuneriM,Hu} the concept of LCD codes and
LCP of codes are generalized and studied over rings, respectively.

Both LCD codes and LCP of codes can be used in cryptography against side channel and fault injection attacks (see
\cite{carlet crypto,ODSM,Bhasin}). In these types of attacks the attacker monitors properties such as power
consumption, electromagnetic leaks and timings of the cryptosystem either passively or actively (by inserting some
faults into the system) and then uses the data gathered to find information about the sensitive data. One
countermeasure against such attacks is data masking. In this method the sensitive data $c$ is added with a random
vector $d$ and the system manipulates $z=c+d$ instead of $c$ itself, so that the attacker cannot get any meaningful
information about $c$. If the space $C$ of sensitive data and the space $D$ of all possible random vectors $d$, form an
LCP of codes, then it is possible to reconstruct the sensitive data $c$ from the masked data $z$. In this case, the
method is called \emph{direct sum masking}. The level of security of this method is shown to be $\min\{\d (C), \d
(D^\perp)\}$, where $\d (C)$ denotes the Hamming distance of $C$, and is called the security parameter $\sc(C,D)$ of
the LCP $(C,D)$ (see \cite{ODSM}).

For any LCP $(C,D)$ of codes we have $\sc(C,D)\leq \d (C)$ and when we choose $C$ to be an LCD code and $D=C^\perp$,
the security parameter is simply $\d (C)$. This raises the question for which codes $C$, there is an LCP of codes
$(C,D)$ such that $\sc(C,D)$ gains the maximum possible value, namely $\d (C)$. One possibility is that there is a code
$D$ equivalent to $C^\perp$ such that $(C,D)$ is an LCP of codes. Such codes are the main subject of this research and
we call them linear complementary equi-dual codes (\good\ codes, for short) (see Definition \ref{main defn}).  In
Section 2, we state several examples and some preliminary results and also we give a necessary condition on $C$ to be
\good. Then in Section 3, we present the conjecture that this necessary condition is also sufficient. We prove several
results which support this conjecture and also study some consequences of trueness of this conjecture.

\section{Preliminary results and a necessary condition}
In the sequel, $F$ always denotes a field and $C$ is a linear code of length $n$ over $F$, that is, a linear subspace
of $F^n$. Moreover, all codes are considered to have length $n$, be linear and over the field $F$ and all matrices are
assume to be over $F$, unless specified otherwise explicitly. The following definition introduces the main concept of
interest in this paper.

\begin{defn}\label{main defn}
We call $C$ a linear complementary equi-dual code, or an \good\ code, if there exists a linear code $D$ with length $n$
over $F$ such that $D$ is permutation equivalent to $C^{\perp}$ and $(C,D)$ is an LCP of codes, that is, $F^n=C\oplus
D$.
\end{defn}

Clearly every LCD code is an \good\ code. Note that one could use more general concepts of equivalence, such as
monomial equivalence, instead of permutation equivalence in the definition of \good\ codes, especially when we consider
fields with more than 2 elements. But if we used monomial equivalence, Conjecture \ref{guess} implies that all codes
would be \good\ when $F\neq F_2$ (see Proposition \ref{consequence}). The following example shows that using
permutation equivalence there exist non-\good\ codes over fields other than $F_2$.
\begin{ex}\label{ex1}
The code $C$ generated by the matrix $ G=\begin{psmat}
1 & 0& 2 \\
0& 1& 2
\end{psmat}$
over $F_3$ is not \good. Indeed, one can find all 6 codes permutation equivalent to $C^\perp$ and check that all have
nonzero intersection with $C$.
\end{ex}

To simplify checking whether a code is \good\ or not, we use the following characterization of LCP of codes. In what
follows $H^t$ denotes the transpose of a matrix $H$. Note that although in \cite{LiuL}, this theorem is stated for
finite fields, but exactly the same proof works for infinite fields.
\begin{thm}[{\cite[Theorem 2.6]{LiuL}}]\label{LCP mat}
Assume that $G$ is a generator matrix for $C$ and $H$ is a parity check matrix for the linear code $D$. Then $(C,D)$ is
an LCP of codes \ifof $\dim (C)+\dim(D) =n$ and $GH^t$ is nonsingular.
\end{thm}

Therefore, we deduce the following characterization of \good\ codes via their generating matrices.
\begin{prop}\label{GC}
Let $G_{k\times n}$ be a generator matrix for $C$. Then $C$  is an \good\   code \ifof there exists a permutation
matrix $P_{n\times n}$ such that $GPG^t$ is invertible.
\end{prop}
\begin{proof}
Note that $D$ is a code permutation equivalent to $C^{\perp}$ \ifof $D$ has a parity check matrix of the form $GP$ or
equivalently $H=GP^t$ for a permutation matrix $P$. Thus $GPG^t=GH^t$ is invertible \ifof $C$ is an \good\ code by
Theorem \ref{LCP mat}.
\end{proof}

Using Proposition \ref{GC}, it is much easier to check that the code $C$ in Example \ref{ex1} is \good. We just have to
check that for all six $3$-square permutation matrices $P$, the matrix $GPG^t$ is not invertible.

\begin{defn}\label{defn mat LCED}
Let $G_{k\times n}$ be a matrix over the arbitrary field $F$. We call $G$ an \good\  matrix, when it is the generator
matrix of an
\good\ code. Equivalently, $G$ is an \good\ matrix, if there exist a permutation matrix $P_{n\times n}$ such that $GPG^t$ is invertible.
\end{defn}
In this paper, we mainly work with matrices and study when a given matrix is \good.
\begin{ex}\label{ex2}
Any full rank matrix $G$ over the field of real numbers is an \good\  matrix, because it is well-known that $GG^t$ is
always invertible.
\end{ex}
\begin{ex}\label{ex3}
Let $G_{4\times 6}$ be the following matrix over $F_2$.

$$
G=\left(
  \begin{array}{cc: cc : cc}
    1 & 0 & 0 & 0 &  1& 1 \\
   0 & 1 & 0 & 0 & 1 & 0 \\\hdashline
    0 &0 & 1& 0 &0 & 1 \\
    0 & 0 & 0 & 1 & 1 & 1 \\
  \end{array}
\right)=
\left(
  \begin{array}{c:c:c}
   I_2 & 0 & A \\\hdashline
    0 & I_2 & B \\
  \end{array}
\right) .$$ We have $$GG^t= \left(
  \begin{array}{cccc}
    1 & 1  &1 & 0 \\
   1 & 0 & 0 & 1 \\
    1 & 0 & 0 & 1 \\
    0 & 1 & 0 & 1 \\
  \end{array}
\right)$$ and $\det (GG^t)=0$. Despite this, $G$ is an \good\  matrix since for the permutation matrix $P$
corresponding to the permutation $(1\ 5)(2\ 6)$, we have
$$ GPG^t=\left(
  \begin{array}{c:c:c}
   I_2 & 0 & A \\\hdashline
    0 & I_2 & B \\
  \end{array}
\right)\left(
  \begin{array}{c:c}
   A^t & B^t \\\hdashline
    0 & I_2  \\ \hdashline
I_2& 0\\
  \end{array}
\right)=\left(
           \begin{array}{cc}
             A+A^t & B^t \\
             B & I_2 \\
           \end{array}
         \right)
.$$ So $\det (GPG^t )=-(\det (B))^2=1$, because $A+A^t=0$.
\end{ex}

Recall that a code $C$ is called $\lambda$-constacyclic for a $\lambda \in F$, when from $(c_1,\ldots, c_n)\in C$, it
follows that $(\lambda c_n, c_1, \ldots, c_{n-1})\in C$. These codes correspond to the ideals of $F[x]/\langle
x^n-\lambda\rangle$. In \cite[Theorem II.4]{CarletG} it is proved that if $(C,D)$ is an LCP of codes and both $C$ and
$D$ are $\lambda$-constacyclic then $D$ is equivalent to $C^\perp$. In fact, they prove that in this case $D$ is the
reciprocal of $C^\perp$ and hence permutation equivalent to $C^\perp$. Similar results are proved for the case that $C$
and $D$ are $n$-dimensional cyclic codes, in \cite[Theorem III.6]{GuneriO}, and more generally for the case that $C$
and $D$ are both $G$-codes for a group $G$, in the main theorem of \cite{Borello}. Thus in all these cases $C$ is a
\good\ code. For the constacyclic case, Theorem II.1 of \cite{CarletG} presents a characterization of LCP $(C,D)$ of codes
with both $C$ and $D$ constacyclic. Combining these two theorems of \cite{CarletG} we deduce the following.

\begin{prop}\label{consta}
A $\lambda$-constacyclic code generated by $g(x)$ over a finite field is \good, if $\gcd(g(x),(x^n-\lambda)/g(x))=1$.
\end{prop}
Note that the converse of this proposition is not correct. To give an example, we need the following results that state
a condition, under which, a cyclic code is \good.  Let $g(x)=\sum_{i=0}^{r} g_i x^i$ be the monic generating polynomial
with degree $r$ for a cyclic $[n,k]$-code $C$ where $n=3r$ and $k=2r$. Then $C$ has a generator matrix of the form
$G=\left(
\begin{array}{c:c:c}
   U_{r\times r} & L_{r\times r} & 0_{r\times r}\\ \hdashline
      0_{r\times r} & U_{r\times r} & L_{r\times r}\\
       \end{array}
        \right)  $, where
        $$ U=\left( \begin{array}{ccccc}
  g_0& g_1 & \ldots & g_{r-2}&g_{r-1} \\
     0& g_0&  \ldots & g_{r-3}&g_{r-2} \\
      \vdots& \vdots & \ddots & \vdots&\vdots \\
       0& 0 & \ldots & g_{0}&g_{1} \\
        0& 0 & \ldots &0&g_{0} \\
       \end{array}
        \right) \ \mathrm{and} \    L=\left( \begin{array}{ccccc}
  g_r& 0 & \ldots & 0&0 \\
     g_{r-1}& g_r&  \ldots & 0&0 \\
      \vdots& \vdots & \ddots & \vdots&\vdots \\
       g_2 & g_3 & \ldots & g_{r}&0 \\
        g_1& g_2 & \ldots & g_{r-1}&g_{r} \\
       \end{array}
        \right) .$$

\begin{lem}\label{25t}
With the above notations if $\chr F=2$, then $ UL^t+LU^t=0$ \ifof  $g(x)g^{\ast}(x)=g_0+g(1)^2x^r+g_0 x^{2r}$, where
$g^\ast(x)=x^rg(\f{1}{x})$.
\end{lem}
\begin{proof}
We have
         $$ (LU^t)^t=UL^t=\left( \begin{array}{ccccc}
  g_0& c_1 & c_2 &\ldots & c_{r-1} \\
    0& g_0& c_1&  \ldots &  c_{r-2} \\
      \vdots& \vdots & \ddots & \vdots&\vdots \\
       0 & 0 & 0 & \ldots & g_{0} \\
       \end{array}
        \right)  ,$$
where $c_i=g_0 g_{r-i}+g_1g_{r-(i-1)}+\cdots +g_ig_r$, for $1\leq i \leq r-1$. Thus $c_i$ is the coefficient of $x^i$
in $g(x)g^{\ast}(x)$. Also note that the coefficient of $x^r$ in $g(x)g^\ast (x)$ is $g_0^2+g_1^2+ \cdots+ g_r^2=
(g(1))^2$. Since $\deg gg^{\ast} =2r$ and $g(x)g^{\ast}(x)$ is self reciprocal, so $c_1=c_2=\cdots=c_{r-1}=0$ \ifof
$g(x)g^{\ast}(x)=g_0+g(1)^2x^r+g_0 x^{2r}$. Also the entries on the main diagonal of $UL^t+LU^t$ are $2g_0=0$, because
$\chr F=2$.
\end{proof}

\begin{prop}\label{reciprocal}
Let $C$ be a cyclic $[n,k]$-code over $F$, $\chr F=2$, and $g(x)=\sum_{i=0}^{r} g_i x^i$ be the monic generating
polynomial of $C$ with $\deg g(x)=r$ and assume that $n=3r$ and $k=2r$. If $g(x)g^{\ast}(x)=g_0+g(1)^2x^r+g_0 x^{2r}$,
then $C$ is an
\good\ code.
\end{prop}
\begin{proof}
Assume that $G, U$ and $L$ are as above. Let $P_{n\times n}$ be the permutation matrix for which
 $GP^t=\left( \begin{array}{c:c:c}
  0 & L & U\\ \hdashline
      L & U & 0\\
       \end{array}
        \right)$
        and so $PG^t=\left( \begin{array}{c:c}
  0 & L^t \\ \hdashline
      L^t & U^t \\ \hdashline
      U^t & 0 \\
       \end{array}
        \right)$.
        Thus
        $$ GPG^t=\left( \begin{array}{c:c}
 L L^t & UL^t+LU^t\\ \hdashline
     UL^t+LU^t & UU^t \\
       \end{array}
        \right)  .$$
By Lemma \ref{25t}, as $g(x)g^{\ast}(x)=g_0+g(1)^2x^r+g_0 x^{2r}$, we have $ UL^t+LU^t=0$. Thus $\det (GPG^t)=(\det
(L))^2(\det (U))^2=g_0^{2r}g_r^{2r}$. Since $g$ is the generating polynomial, $g_0,g_r\neq 0$ and $GPG^t$ is
invertible. Consequently, $C$ is an \good\ code.
\end{proof}
\begin{ex}\label{exre}
Let $C=\langle\langle x^4 +x^2 +1 \rangle\rangle $ be a cyclic code of length 12 over $F_2$. We have $g(x)g^{\ast}(x)=
(x^4+x^2 +1)^2 = x^8+x^4 +1$. We have $g(x)g^{\ast}(x)=g_0+g(1)^2x^r+g_0 x^{2r}$. By Proposition \ref{reciprocal}, $C$
is an \good\ code. Note that $g(x)^2|x^{12}-1$ and hence $\gcd(g(x), (x^{12}-1)/g(x))= g(x)\neq 1$. Therefore, there is
no LCP of codes $(C,D)$ such that $D$ is also cyclic according to \cite[Theorem II.1]{CarletG}.
\end{ex}

Next, we study some elementary properties of \good\ matrices.
\begin{lem}\label{3e}
Let $G_{k\times n}$ be a matrix. The following are equivalent.
\begin{enumerate}
  \item $G$ is an \good\  matrix.
  \item The matrix $EG$ is \good, where $E$ is an invertible matrix.
  \item The matrix $GP$ is \good, where $P$ is a permutation matrix.
\end{enumerate}
In this case, 
$\rank G=k$.
\end{lem}
\begin{proof}
(i$\leftrightarrow$ii) For the invertible matrix $E$, the matrix $EG$ is \good\  \ifof there exists the permutation
matrix $P'_{n\times n}$ such that $(EG)P'(EG)^t$ is invertible \ifof $E(GP'G^t)E^t$ is invertible \ifof $GP'G^t$ is
invertible \ifof $G$ is an \good\  matrix.

(iii$\rightarrow $i) Assume that $P$ is a permutation matrix and $GP$ is
\good. Thus there exists a permutation matrix $P'$ such that $(GP)P'(GP)^t=G(PP'P^t)G^t$ is invertible. Since $PP'P^t$
is a permutation matrix, $G$ is an \good\  matrix.

(i$\rightarrow $iii) Suppose that $G$ is an \good\  matrix. For any permutation matrix $P$, the matrix $(GP)P^{-1}$ is
\good. Since $P^{-1}$ is a permutation matrix, $GP$ is \good\ by (iii$\rightarrow$i) above.

For the final statement, if $\rank G<k$, then for any permutation matrix $P$ we have $\rank (GPG^t)\leq \rank G < k$.
So $GPG^t$ is not invertible and  $G$ is not \good.
\end{proof}
We know that every matrix $G_{k\times n}$ with rank $k$ can be transformed to a \emph{standard form} as $G'=( I_k | A)
_{k\times n}$ which is obtained by  applying a permutation on the columns of $G$ and some row-elementary operations on
$G$, that is, $G'=EGP$ for a permutation matrix $P$ and an invertible matrix $E$. Thus from Lemma \ref{3e} we
immediately deduce the following.
\begin{cor}\label{c1}
Let $G_{k\times n}$ be a matrix with rank $k$. Then $G$ is an \good\  matrix \ifof  a standard form of $G$ is \good.
\end{cor}
\begin{cor}\label{c2}
Let $C$ be an \good\ code. Then any code permutation equivalent to $C$ is \good.
\end{cor}

Note that by Lemma \ref{3e} if $G_{k\times n}$ is \good, then $k\leq n$. From now on, we assume that $k\leq n$ and by
$G$ we always denotes a $k\times n$ matrix.
\begin{lem}\label{3eg}
Let $A$ be a $k\times (n-k)$ matrix. The following are equivalent.
\begin{enumerate}
  \item The matrix $(
I_k | A)$ is \good.
  \item The matrix $
   (I_k | PA)$ is \good, where $P$ is a $k\times k$ permutation matrix.
  \item The matrix $ (
   I_k | AP') $ is \good, where $P'$ is a $(n-k)\times (n-k)$ permutation matrix.
\end{enumerate}
\end{lem}
\begin{proof}
Set $G=(
   I_k | A )$.\\
  (i$\leftrightarrow$ii) Since $P_{k\times k}$ is a permutation matrix, the matrix $Q=\left(
   \begin{array}{c:c}
    P^{-1} & 0 \\\hdashline
     0 & I_{ n-k}\\
     \end{array}
     \right)_{n\times n}$ is also a permutation matrix. We have
     $$(PG)Q
=(
 P | PA)
      \left(
     \begin{array}{c:c}
      P^{-1} & 0 \\\hdashline
     0 & I_{n- k}\\
     \end{array}
     \right)=
     (
      I_k | PA).$$
       By Lemma \ref{3e}, $G$ is \good\  \ifof the above matrix is \good.\\
  (i$\leftrightarrow$iii) Since $P'_{(n-k)\times (n- k)}$ is a permutation matrix, the matrix $\left(
   \begin{array}{c:c}
    I_{ k} & 0 \\\hdashline
     0 & P'\\
     \end{array}
     \right)_{n\times n}$ is also a permutation matrix. We have
     $$ (
   I_k | AP')=
    (
   I_k | A)
    \left(
     \begin{array}{c:c}
      I_{ k}  & 0 \\\hdashline
     0 & P'\\
     \end{array}
     \right).$$
      By Lemma \ref{3e}, $G$ is \good\  \ifof the above matrix is \good.
\end{proof}

\begin{rem}
Let $G=(
   I_k | A )_{k\times n}$ and $Q$ be an $(n-k)$-square permutation matrix. If $P=\left(
     \begin{array}{c:c}
      I_{ k}  & 0 \\\hdashline
     0 & Q\\
     \end{array}
     \right)_{n \times n}$, then $GPG^t= I_k+AQA^t$. Thus we have the following.
  \begin{enumerate}
    \item If $AQA^t=\lambda I_k$, then $GPG^t$ is a multiple of the identity matrix and if $\lambda \neq -1$, then
        $G$ is \good.
    \item If $AQA^t$ is a nilpotent matrix, then $G$ is an \good\  matrix. In particular, if  $AA^t$ is a nilpotent
        matrix, then $G$ is an \good\  matrix.
  \end{enumerate}
\end{rem}

In the sequel, we denote permutation equivalence of codes by $\simeq$.
\begin{prop}\label{c2d}
If $C$ is an \good\  code, then $C^{\perp}$ is also an \good\  code.
\end{prop}
\begin{proof}
Suppose that $(C,D)$ is an LCP of codes and $D\simeq C^{\perp}$. Thus $C+D=F_q^n$ and $C\cap D={0}$. So
$(C^{\perp})^{\perp} \simeq D^{\perp}$,
$${0}=(F_q^n)^{\perp}=(C+D)^{\perp}=C^{\perp}\cap D^{\perp}$$
and
$$F_q^n=({0})^{\perp}=(C\cap D)^{\perp}=C^{\perp}+ D^{\perp}.$$
\end{proof}

\begin{lem}\label{gt}
The matrix $G_1=(I_k | A )_{k \times n}$ is \good\  \ifof $G_2=(
 I_{(n-k)} | -A^t )_{(n-k) \times n}$ is an \good\  matrix.
\end{lem}
\begin{proof}
$G_1$ is a generator matrix for a code $C$ \ifof $G_2$ is a generator matrix for $C^{\perp}$. Now the result follows
from Proposition \ref{c2d}.
\end{proof}

At the end of this section, we present an easy necessary condition on a matrix to be \good.
\begin{thm}\label{011}
If $G$ does not have full rank or if the sum of the entries on any row of $G$ is $0$ and  $(1,1,\ldots,1)$ is in the
row space of $G$, then $G$ is not
\good.
\end{thm}
\begin{proof}
The claim on the rank of $G$ was proved in Lemma \ref{3e}. Suppose that $vG=(1,\ldots,1)$ for a nonzero vector $v$. For
every permutation matrix $P$, we have
  $$GPG^tv^t=GP\left(
           \begin{array}{c}
             1 \\
             1\\
            \vdots \\
            1 \\
           \end{array}
         \right)=
         G\left(
           \begin{array}{c}
             1 \\
             1\\
            \vdots \\
            1 \\
           \end{array}
         \right)=0.
   $$
So $GPG^t$ is not invertible.
\end{proof}

Note that if $G$ is in standard form and satisfies the condition of the previous theorem, then the vector $v$ in the
proof must be $(1,\ldots,1)$, that is, the sum of the entries on any column of $G$ is equal to 1. This is exactly the
case for the matrix $G$ in Example \ref{ex1}.

\section{A conjecture}

We used CoCoA computer software (\cite{cocoa}) to find non-\good\ matrices such as $G_{k\times n}$ with $k\leq 3$  and
$n\leq 7$ and over finite fields $F_p$ where $p<20$ is a prime. We observed that all such matrices satisfy the
conditions of Theorem \ref{011}. Using Lemma \ref{gt}, this can be extended to all $k\leq 7$. Thus we pose the
following conjecture, and throughout this section, we show it is correct in several cases.

\begin{conj}\label{guess}
We conjecture that the converse of Theorem \ref{011} is true. In particular, if $G$ is in standard form and is not
\good, then the sum of the entries on any row of $G$ is $0$ and the sum of the entries on any column of $G$ is $1$.
\end{conj}

Note that the ``in particular'' statement in the conjecture implies the whole conjecture, because of Lemma \ref{c1}.

\begin{rem} \label{l2} Suppose that $G$ is a generating matrix for a code $C$ and $e=(1,\ldots,1)_{1\times n}$. Then the sum of the
entries on any row of $G$ is $0$ \ifof $Ge^t=0$, that is, $e\in C^\perp$. Also if $G$ is in standard form, then the sum
of the entries on any column of $G$ is $1$, \ifof $e$ is in the row space of $G$, that is, $e\in C$. Thus in the
language of coding theory, Conjecture \ref{guess} reads as ``$C$ is not \good\ \ifof $e\in C\cap C^\perp$''.
\end{rem}

Before proving Conjecture \ref{guess} in some special cases, we mention some consequences of trueness of this
conjecture.

\begin{prop}\label{consequence}
Suppose that Conjecture \ref{guess} holds. Then
\begin{enumerate}
\item If $\chr F\nmid n$, then all length $n$ codes are \good;

\item If $q>2$, then for every $q$-ary code $C$ there exists a code $D$ such that $D$ is monomially equivalent to
    $C^\perp$ and $(C,D)$ is an LCP of codes.
\end{enumerate}
\end{prop}
\begin{proof}
(i): Suppose that $C$ is a non-\good\ code with length $n$. Using Lemma \ref{c2}, we can assume that $C$ has a
generating matrix in standard form such as $G$. Now according to the conjecture the sum of all entries of $G$ is
$(k)(0)=(n)(1_F)$. Thus $n$ must be zero in $F$, that is, $\chr F|n$, a contradiction.

(ii): If $C$ is \good, then clearly the result holds. Thus suppose that $C$ is not \good. Using a monomial version of
Lemmas \ref{3e} and \ref{c2} (where we just replace permutation matrices with monomial matrices), we can assume that
$C$ has a generating matrix $G$ in standard form. Suppose that $0,1\neq a\in F$ and let $A$ be the diagonal matrix
$\diag(a,1,\ldots, 1)$ and $G'=GA$. Then the sum of entries on the first column of $G'$ is different from the sum of
entries on any other column of $G'$. Hence according to the conjecture, $G'$ is
\good. Therefore, there exists a permutation matrix $P$ such that $G'PG'^t=GAPAG^t$ is invertible. Let $M=AP^tA$ which
is a monomial matrix and $D$ be the code with the parity check matrix $GM$. Then $D$ is monomially equivalent to $C$
and it follows from Theorem \ref{LCP mat} that $(C,D)$ is an LCP of codes.
\end{proof}

In the sequel, we present some results supporting correctness of the conjecture and prove that it holds in several
special cases. First we consider the case that $n=k+1$. Hereafter by $S_k$ we mean the symmetric group on $\{1,\ldots,
k\}$  and for any permutation, $\sigma \in S_k$, its corresponding permutation matrix is denoted by $P_{\sigma }$.
Indeed, if $\delta$ is the Kronecker delta,
 $$(P_{\sigma})_{ij}= \delta_{\sigma(j),i}= \left\{ \begin{array}{ccc}
       1&\ \ \ & \sigma(j)=i\\
       0&\ \ \ & \sigma(j)\neq i \\
      \end{array}\right. .$$
\begin{lem}\label{l3}
Assume that  $\pi\in S_{k+1}$ with $\pi(k+1)\neq k+1$, $P=P_\pi$ and 
$$G=\left(
\begin{array}{c:c}
I_k &
         \begin{array}{c}
           a_1 \\
           a_2 \\
           \vdots \\
           a_k \\
         \end{array}
 \\
 \end{array}
 \right)_{k\times (k+1)}.$$
Then $$\det(GPG^t)=\pm \left(a_j +a_i -\sum^k _{r=1, r\neq j} a_{\pi(r)}a_r\right),$$ where $i=\pi(k+1)$ and
$j=\pi^{-1}(k+1)$.
\end{lem}
\begin{proof}
Suppose that $B=GPG^t$. It can be seen that $B_{ij}=a_i+a_j$, $B_{il}=a_l$ for each $l\neq j$, $B_{rj}=a_r$ for each
$r\neq i$ and $B_{rl}=\delta_{r,\pi(l)}$ for each $l\neq j$ and $r\neq i$. Apply the following elementary operations on
$B$: for all $r\neq j$, add $-a_{\pi(r)}$ times of the $r$-th column to the $j$-th column. We get a matrix in which the
$(i,j)$-entry is $d_{\pi}=a_j +a_i -\sum^k _{r=1 , r\neq j} a_{\pi(r)}a_r$ and other entries on the $j$-th column are
$0$. On each row of this matrix other than row $i$ and each column of this matrix other than column $j$, there exists
exactly one $1$ and other entries are $0$. From this, the result follows.
\end{proof}

\begin{thm}\label{g4}
  Let $n=k+1$ and 
 $$G=\left(
\begin{array}{c:c}
I_k &
         \begin{array}{c}
           a_1 \\
           a_2 \\
           \vdots \\
           a_k \\
         \end{array}
   \end{array}
 \right)_{k\times (k+1)}.$$
Then $G$ is not \good\ \ifof $\chr F\mid n$ and $a_i=-1$, for all $i$. In other words, Conjecture \ref{guess} is true
in this case.
\end{thm}
\begin{proof}
One side follows from Theorem \ref{011}. Assume that $G$ is not \good. If $k=1$, the statement is clear. Suppose $k\geq
2$. Then for any permutation $\pi$, where $\pi(k+1)= i\neq k+1$ and $\pi(j)=k+1$, by Lemma \ref{l3}, we have
\begin{align}\label{p}
 e_\pi=\pm \det(GP_{\pi}G^t)= a_j +a_i -\sum^k _{r=1, r\neq j}a_{\pi(r)}a_r =0.
\end{align}
For each $r_1\neq r_2\leq k$ consider the permutations $\pi _1=(r_1\ k+1)$, $\pi_2=(r_2\ k+1)$ and $\pi_3=(r_1\ r_2\
k+1)$. Then we get
$$2e_{\pi _3}-e_{\pi_1}-e_{\pi _2}=a_{r_1}^2+a_{r_2}^2-2a_{r_1}a_{r_2}=0.$$
Hence $(a_{r_1}-a_{r_2})^2=0$ and $a_{r_1}=a_{r_2}$, for all $r_1\neq r_2$. Let $a=a_i$ for some or equivalently all
$i$. Then we have $GG^t=I_k+a^2U$, where $U$ is a $k$-square matrix with all entries 1. If first, we consecutively
subtract row $i$ of $GG^t$ from row $i-1$ for $i=2,3,\ldots,k$ and then subtract $ia^2$ times the $i$-th row from the
last row for $i=1,\ldots,k-1$, we obtain the following matrix:
$$B=\begin{pmatrix}
1 & -1 & 0 & 0 & \ldots & 0\\
0 & 1 & -1 & 0 &\ldots &0 \\
\vdots & \vdots & \ddots & \ddots & \ddots & \vdots \\
0 & 0&  \ldots & 0 & 1 & -1 \\
0 & 0 & \ldots & 0 & 0 & 1+ka^2
\end{pmatrix}.$$
Hence $0=\det(GG^t)=\det(B)=1+ka^2$, that is, $ka^2=-1$. Now if in \eqref{p} we replace all $a_l$'s with $a$ and then
replace $ka^2$ with $-1$, we get $2a+1+a^2=0$, which implies $a=-1$. Now $1+ka^2=0$ shows that $k+1=0$, that is, $\chr
F| n$, as required.
\end{proof}

Now we can complete Example \ref{ex1}.
\begin{ex}
By Theorem \ref{g4}, all $2\times 3$ standard matrices are \good, except in the case that $\chr F=3$ and the matrix is
$\begin{psmat} 1& 0& 2 \\
0 & 1& 2 \end{psmat}$
\end{ex}

The following states another case in which the conjecture holds.
\begin{cor}\label{guess2}
For any $1\times n$ standard matrix Conjecture \ref{guess} holds.
\end{cor}
\begin{proof}
Suppose that $G=( 1 | a_1 , \ldots , a_{n-1})_{1\times n}$
is not \good. By Lemma \ref{gt}, the matrix $$G'=\left(
  \begin{array}{c:c}
    I_{n-1} & \begin{array}{c}
    -a_1 \\
     \vdots \\
     -a_{n-1}
     \end{array}
      \\
      \end{array}
      \right)$$
is not an \good\ matrix. Hence by Theorem \ref{g4}, $\chr F\mid n$ and $-a_i=-1$, that is, $a_i=1$ for all $i$ and the
claim follows.
\end{proof}

To find other conditions under which the conjecture holds, we use the following observation.
\begin{lem}\label{eigen}
Suppose that $G=(I_k|A)$ is a non-\good\ matrix, $Q$ is a $(n-k)$-square permutation matrix and $M=AQA^t$. Then $-1$ is
an eigenvalue of $M$ and all of its row (or equivalently, column) permutations.
\end{lem}
\begin{proof}
Let $P'$ be an arbitrary $k$-square permutation matrix and
 $$P=\left(\begin{array}{c:c}
 P' & 0 \\
 \hdashline
 0 & Q
 \end{array}\right).$$
Then $0=\det(GPG^t)=\det(P'+APA^t)$. Multiplying by $P'^{-1}$, we get $\det(I+P'^{-1}M)=0$ (or $\det(I+MP'^{-1})=0$).
Noting that $P'$ and hence $P'^{-1}$ are arbitrary, the proof is concluded.
\end{proof}

To utilize the above statement, we first study matrices $M$ with the property that $-1$ an eigenvalue of $PM$ for every
permutation matrix $P$.

\begin{defn}
We say that a field $F$ satisfies condition $\Pi _k$, when for any $k\times k$ symmetric matrix $M$ over $F$, if $PM$
has $-1$ as an eigenvalue for every permutation matrix $P$, then the sum of all entries of $M$ is $-k$.
\end{defn}

Below we will show that many fields satisfy condition $\Pi_k$ (see Corollary \ref{c4dd}), but first we state the
following proposition which shows why we are interested in fields satisfying condition $\Pi_k$.

\begin{prop}\label{alpha}
Suppose that $k\geq 2$ and $G=( I_k|  A)$
is a non-\good\ matrix.
\begin{enumerate}
  \item  If condition $\Pi _k$ holds for $F$, then there is a constant $\beta\in F$ satisfying  $(n-k)\beta ^2 =-k$
      such that the sum of all entries on any column of $A$ is $\beta$.
  \item If $F$ satisfies condition $\Pi _{n-k}$, then there is a constant $\alpha\in F$ satisfying  $k\alpha ^2
      =-(n-k)$ such that the sum of all entries on any row of $A$ is $\alpha$.
\end{enumerate}
\end{prop}
\begin{proof}
(i): Let $A=(a_{ij})$ and $M_{k\times k}$ be the symmetric matrix $AA^t$. By Lemma \ref{eigen}, for any permutation
matrix $P$,
      $\det (I+PM)=0$. Since condition $\Pi _k$ holds, $\sum_{i,j=1}^k
      M_{ij}=-k$. We know that $M_{ij}=\sum_{t=1}^m a_{it} a_{jt}$, where $m=n-k$. Now,
\begin{equation}\label{alpha1}
-k=\sum_{i,j=1}^k M_{ij}=\sum_{i,j=1}^k\sum_{t=1}^m a_{it} a_{jt}=\sum_{t=1}^m(\sum_{j=1}^ka_{jt})(\sum_{i=1}^ka_{it} )=\sum_{t=1}^m \beta_t^2,
\end{equation}
where $\beta _t$ is the sum of all entries on column $t$ of $A$.

Let $Q=P_{(1\ 2)}$ and $D=AQA^t$. Again by \ref{eigen}, $PD$ has an eigenvalue $-1$ for all permutation matrices $P$.
Since $Q=Q^{-1}=Q^t$, it follows that $D$ is a symmetric matrix and hence $\sum_{i,j=1}^k D_{ij}=-k$. Now
\begin{equation*}
 D_{ij}= \sum_{t=1}^m a_{it}(QA^t)_{tj}=(\sum_{t=3}^m a_{it}a_{jt})+(a_{i1}a_{j2}+a_{i2}a_{j1})
\end{equation*}
and
\begin{align}\label{alpha2}
  -k & = \sum_{i,j=1}^k D_{ij}=(\sum_{i,j=1}^k\sum_{t=3}^m a_{it} a_{jt})+(\sum_{i,j=1}^k  a_{i1}a_{j2} )+ (\sum_{i,j=1}^k a_{i2}a_{j1} ) \notag\\
   & = ( \sum_{t=3}^m \beta _t^2)+2\beta _1 \beta _2.
\end{align}
Thus by \eqref{alpha1} and \eqref{alpha2}, $\beta_1 ^2 +\beta_2 ^2 = 2 \beta_1 \beta_2$. So $\beta_1 = \beta_2$.
Similarly, $\beta_1 =\beta_2= \cdots =\beta_m =\beta$. Also, $-k=m\beta^2 $, by \eqref{alpha1}.

(ii): By Lemma \ref{gt}, $G'=(I_{(n-k)} | -A^t)_{(n-k)\times n}$ is not an \good\ matrix. Now, the result is obtained by
applying part (i) on $G'$.
\end{proof}

As an example of utilizing this proposition, we have the following.
\begin{cor}\label{char2 odd k,n-k}
Suppose that $k$ is odd and $n$ is even. Also, let $F$ be a field  with characteristic $2$ which satisfies $\Pi _k$ and
$\Pi_{n-k}$. If $( I_k |  A)_{k\times n}$
is not an \good\ matrix, then the sum of all entries on each row or each column of $A$ is $1$. That is, Conjecture
\ref{guess} is true in this case.
\end{cor}
\begin{proof}
Note that in $F$ we have $k=n-k=1$. Therefore, it follows from Proposition \ref{alpha}, that $\alpha ^2 =1$ and hence
$\alpha =1$. Similarly, $\beta =1$.
\end{proof}

Note that when $\chr F=2$, $k$ is odd and $M_{k\times k}$ is a symmetric matrix, showing that the sum of all  entries
of $M$ is $-k$ is equivalent to showing that $\tr(M)=1$.
\begin{lem}\label{33}
 Let $M=\begin{psmat}
            a & x & y \\
            x & b & z \\
           y & z & c \\
          \end{psmat}$
be a symmetric matrix over a field $F$ with $\chr F=2$. If $-1$ is an eigenvalue of $PM$ for every permutation matrix
$P$, then $tr(M)=1$.
\end{lem}
\begin{proof}
Set $d=\det (M)$ and $t=tr(M)$. For any $3\times 3$ matrix $B=(b_{ij})$, we have $\det(I+B)=0$ \ifof
\begin{equation}\label{B}
1+\det (B)+\tr(B)=(b_{11}b_{22}+b_{11}b_{33}+b_{22}b_{33})+(b_{12}b_{21}+b_{13}b_{31}+b_{23}b_{32}).
\end{equation}
Writing \eqref{B} for $B=P_\pi M$, where $\pi$ is either $ (1\ 2)$, $(1\ 3)$, $(2\ 3)$ or $(1\ 2\ 3)$, we deduce that
\begin{eqnarray}
  1+d+c &=& ab+x^2,  \label{1} \\
  1+d+b &=& ac+y^2,  \label{2}  \\
  1+d+a &=& bc+z^2,  \label{3} \\
 1+d+(x+y+z) &=& az+cx+by+xy+xz+yz. \label{4}
\end{eqnarray}
If we add $c$ times of \eqref{1} with $b$ times of \eqref{2} and $a$ times of \eqref{3}, we have
$$(a+b+c)(1+d)+(a+b+c)^2 = abc+cx^2+by^2+az^2.$$
The right hand side of the above equation is exactly $d$, so $(1+d)t+t^2+d=0$. Thus $t=1$ or $t=d$. If $t=d$, then by
replacing $d=a+b+c$ in \eqref{1}, \eqref{2} and \eqref{3} we have,
\begin{eqnarray}
  (1+a)(1+b) &=& x^2,  \label{5} \\
  (1+b)(1+c) &=& z^2,  \label{6}  \\
 (1+a)(1+c) &=& y^2.  \label{7}
\end{eqnarray}
If we multiply two of the above equations by each other and simplify the result using the third equation, we get
\begin{eqnarray*}
  x^2 y^2 &=& (1+a)^2 z^2 = z^2+z^2a^2,   \\
   x^2 z^2 &=& (1+c)^2 x^2 = x^2+x^2c^2,   \\
  z^2 y^2 &=& (1+b)^2 y^2 = y^2+y^2b^2 .  \\
\end{eqnarray*}
By replacing the above values in the square of \eqref{4}, we have  $d^2=1$ and so $d=1$.
\end{proof}
Note that Lemma \ref{33} indeed states that $\Pi_3$ holds for fields with characteristic 2. Thus from Corollary
\ref{char2 odd k,n-k}, we immediately get the following.
\begin{cor}\label{gu2}
If $\chr F=2$, $n=6$ and $k=3$, then  Conjecture \ref{guess} is true.
\end{cor}

Next we show that $\Pi_2$ holds for every field and use it to show that the conjecture holds for $k=2$. For this, we
first investigate the summation of characteristic polynomials of a matrix and all of its permutations. In the sequel we
denote the $k$-square matrix with all entries 1 by $U$. Also by $A_k$ and $B_k$ we mean the set of even and odd
permutations in $S_k$, respectively.
\begin{lem}\label{psigma}
 \begin{enumerate}
  \item For $k\geq 3$ we have $$\sum _{\sigma \in A_k} P_{\sigma}= \sum _{\sigma \in B_k} P_{\sigma}=
      \frac{(k-1)!}{2}U\text{ and }\sum _{\sigma \in S_k} (\det (P_{\sigma}))P_{\sigma}= 0.$$
  \item For $k\geq 1$ we have $$\sum _{\sigma \in S_k} P_{\sigma}= (k-1)!U .$$
\end{enumerate}
\end{lem}
\begin{proof}
(i) Set $A=\sum _{\sigma \in A_k} P_{\sigma}$ and $B=\sum _{\sigma \in B_k} P_{\sigma}$. Thus $A_{ij}=\sum _{\sigma \in
A_k} (P_{\sigma})_{ij}$ and is equal to the number of even permutations that map $j$ onto $i$. It is easy to see that
the number of all permutations that map $j$ onto $i$ (and exactly half of which is even) is $(k-1)!$ and hence
$A_{ij}=\frac{(k-1)!}{2}$. Hence $A=\frac{(k-1)!}{2}U$. Similarly, $B=\frac{(k-1)!}{2}U$. Using the fact that $\det
(P_\sigma)$ is the sign of $\sigma$ for every $\sigma\in S_k$, the other claim follows easily.

(ii) If $k\leq 2$, the result is clear and if $k\geq 3$ the claim follows part (i).
\end{proof}

Let $A$ be a $k$-square matrix and $\{j_1,j_2,\ldots ,j_r\}$ be a subset of $[k]=\{1,2,\ldots ,k\}$. We use
$A_{\{j_1,j_2,\ldots ,j_r\}}$ to denote the $r\times r$ submatrix of $A$ whose rows and columns are rows and columns
$j_1,j_2,\ldots ,j_r$ of $A$, respectively.
Also by $\delta _r(A)$ we mean $\sum \det (A_{\{j_1,j_2,\ldots ,j_r\}})$, where the summation is taken over all
$r$-subsets $\{j_1,j_2,\ldots ,j_r\}$ of $[k]$. We denote the family of all $r$-subsets of $[k]$ by $\binom{[k]}{r}$.
Note that $\delta _1(A)= \tr (A) $, $\delta _k(A)=\det (A)$ and $\delta _r(I_k)=\binom{k}{r}$. More generally, we have
the following.
\begin{thm}[{\cite[Theroem 6.164]{Comb Fund Alg}}]\label{CA}
Let $A$ be a $k$-square matrix. Then the characteristic polynomial of $A$ is
$$C_A(x)=x^k -\delta _1(A)x^{k-1}+\delta _2(A)x^{k-2}+\cdots +(-1)^{k-1}\delta _{k-1}(A)x+(-1)^{k}\delta _{k}(A).$$
\end{thm}
In what follows we give formulations for the sum of the coefficients of the characteristic polynomials of a matrix and
all of its row (or column) permutations.

\begin{thm}\label{delta2}
Suppose that $k\geq 2$. Let $A=(a_{ij})$ be a $k\times k$ matrix. Then $$\sum_{\pi \in S_k} \delta _1(P_\pi
A)=(k-1)!\l(\sum _{i,j=1}^k a_{ij}\r).$$
\end{thm}
\begin{proof}
We have
\begin{eqnarray*}
         \sum_{\pi \in S_k} \delta _1(P_\pi A) &=&\sum_{\pi \in S_k} \tr(P_{\pi} A)=\tr\l(\l(\sum_{\pi \in S_k} P_{\pi}\r) A\r)  \\
         &=& \tr((k-1)!UA) \hfill \qquad \text{ (by Lemma \ref{psigma})}\\
& = &(k-1)!\ \tr(UA) \\
 &= &(k-1)!\ \l(\sum _{i,j=1}^k a_{ij}\r),
\end{eqnarray*}
as required.
\end{proof}
To get a similar result for $\delta_r$ when $r>1$ we use the following lemma whose proof is easy and left to the
reader.
\begin{lem}\label{correspondence}
Let $(a_i)_{i\in I}$ be a sequence of elements of $F$ with $|I|<\infty$. If there exists a function $f:I\to I$ such
that $f^2=\id_I$ and $a_i=-a_{f(i)}$ and $f(i)\neq i$ for each $i\in I$, then $\sum_{i\in I} a_i=0$.
\end{lem}
\begin{thm}\label{delta}
If $A$ is a $k\times k$  matrix and $2\leq r\leq k$, then $\sum_{\pi \in S_k} \delta _r(P_\pi A)=0$.
\end{thm}
\begin{proof}

Let $I=\left\{(\pi,J)|\pi\in S_k, J\in \binom{[k]}{r}\right\}$. For each $i=(\pi,J)\in I$, set $a_i=\det (P_\pi A)_J$.
Then $\sum_{\pi \in S_k} \delta _r(P_\pi A)=\sum_{i\in I}a_i$. As $|J|=r\geq 2$, we can choose the two smallest
elements of $J$, say $j_1, j_2$. Define $f:I\to I$ by $f(i)=f(\pi, J)=((j_1\ j_2)\pi, J)\neq i$. Then $f^2=\id_I$ and
$\left(P_{(j_1\ j_2) \pi}A\right)_J$ is the matrix obtained by replacing the first and second rows of $(P_\pi A)_J$
with each other. Thus $a_i=-a_{f(i)}$ and the claim follows from Lemma \ref{correspondence}.
\end{proof}

\begin{cor}\label{c1dd}
Let $A=(a_{ij})$ be a $k\times k$ matrix. Then $$\sum_{\pi \in S_k} C_{P_\pi}(x)
=k!\ x^k-(k-1)!\left(\sum _{i,j=1}^k a_{ij}\right)x^{k-1}.$$
\end{cor}
\begin{proof}
It is clear by Theorems \ref{CA}, \ref{delta} and \ref{delta2}.
\end{proof}

Recall that, if $\pi \in S_k$ decomposes into disjoint cycles as
$$ \pi =(a_1 a_2 \ldots a_{i_0})(a'_1 a'_2 \ldots a'_{i_1})\cdots(a^{(t)}_1 a^{(t)}_2 \ldots a^{(t)}_{i_t}),
$$
then the characteristic polynomial of $P_{\pi}$ is
$$C_{P_{\pi}}(x)=(x^{i_{0}}-1)(x^{i_{1}}-1)\cdots(x^{i_{t}}-1)(x-1)^{k-\sum _{j=0}^t i_j}.$$
For example, if $\pi=(12)(34)(567)\in S_8$, then $C_{P_{\pi}}(x)=(x^2-1)^2 (x^3-1)(x-1).$
\begin{cor}\label{c2dd}
The sum of characteristic polynomials of all $k\times k$ permutation matrices is $k!(x^k-x^{k-1}).$
\end{cor}
\begin{proof}
Apply Corollary \ref{c1dd} with $A=I_k$.
\end{proof}

\begin{cor}\label{c3dd}
Let $M$ be a $k\times k$ matrix over a field $F$ such that $-1$ is an eigenvalue of $PM$ for all permutation matrices
$P$. Then $$(k-1)! \ (k+\sum _{i,j=1}^k M_{ij})=0.$$
\end{cor}
\begin{proof}
Since $-1$ is a root of the characteristic polynomial of $PM$, for any permutation matrix $P$,  by Corollary
\ref{c1dd}, $0=\sum_{\pi \in S_k} \det((-1)I_k-P_\pi M)=k!\ (-1)^k-(k-1)!\ (\sum _{i,j=1}^k M_{ij})(-1)^{k-1}.$ Hence
$(k-1)!(k+\sum _{i,j=1}^k M_{ij})=0$.
\end{proof}

Note that if $(k-1)!\neq 0$ in $F$, that is, if $\chr F\geq k$, then in the above lemma we could deduce that the sum of
all entries of $M$ is $-k$. This is stronger than condition $\Pi_k$, since $M$ here is not assumed to be symmetric.

\begin{cor}\label{c4dd}
If $\chr F=0$ or if $\chr F\geq k$, then condition $\Pi_k$ holds for $F$. In particular, every field satisfies $\Pi_1$,
$\Pi_2$ and $\Pi_3$.
\end{cor}
\begin{proof}
This is an immediate consequence of Corollary \ref{c3dd}. For $\Pi_3$ note that by Lemma \ref{33}, this condition holds
for fields with characteristic 2.
\end{proof}

\begin{cor}\label{char=0}
Suppose that $F$ is a field with characteristic 0 in which $\f{-k}{n-k}$ is not a square. Then every $[n,k]$-code over
$F$ is \good.
\end{cor}
\begin{proof}
Just note that by Corollary \ref{c4dd}, $F$ satisfies $\Pi_k$ and hence according to Proposition \ref{alpha} if there
exists a non-\good\ $[n,k]$-code over $F$, then $\f{-k}{n-k}$ is a square in $F$.
\end{proof}

By an argument similar to the above proof one can find positive integers $n,k$ and prime numbers  $p$ such that every
$[n,k]$-code over $F_p$ is \good. For this we recall the\emph{ Legendre symbol} $\tleg{a}{p}$, for $a\in \z$ with
$\gcd(a,p)=1$ which is 1 if $a$ is a quadratic residue modulo $p$ and $-1$ otherwise (see, for example, \cite[Section
8.4]{roman} for properties of $\tleg{a}{p}$).

\begin{ex}
Suppose that $n=tk$ for some $t\in \n$, $p$ is a prime number with $2\leq k<p$ and $F=F_{p^m}$ where $m$ is odd.  In
each of the following cases every $[n,k]$-code over $F$ is \good:
\begin{itemize}
\item $p \equiv -1 \text{ or } -3 \mod 8$ and $t=3$;
\item $p \equiv  -1\text{ or } -3 \text{ or } -7 \mod 12$ and $t=4$;
\item $p\equiv -1 \mod 4$ and $t\equiv a^2+1 \mod p$ for some $a\in \z$ with $p\nmid a$.
\end{itemize}
The reason is that in each of these cases $n-k\neq 0$ in $F$ and since by Corollary \ref{c4dd}, $F$ satisfies condition
$\Pi_k$, according to Proposition \ref{alpha}, $\f{-k}{n-k}$ must be a square in $F$. Noting that if an element is not
a square in a field $F_p$, then it is not a square in any finite dimensional extension of $F_p$ with odd dimension such
as $F$, it suffices to show that in each of these cases $\f{-k}{n-k}$ is a quadratic nonresidue modulo $p$. We have
$$\leg{-k/n-k}{p}=\leg{-(n-k)/k}{p}=\leg{-t+1}{p}=\leg{-1}{p}\leg{t-1}{p}=-1,$$
where the last equality follows from \cite[Theorem 8.4.2]{roman} in all of the three cases. It worths mentioning that
this example shows that in the aforementioned cases Conjecture \ref{guess} holds.
\end{ex}

It should be mentioned that the authors could not find any field $F$ and positive integer $k$ such that $F$ does not
satisfy $\Pi_k$ or even the stronger condition in which the word ``symmetric'' is removed from the definition of
$\Pi_k$. Thus we propose the following Conjecture.

\begin{conj}\label{conj Pk}
Every field satisfies $\Pi_k$ for every positive integer $k$. More generally, for any $k\times k$ matrix $M$ (not
necessarily symmetric) over $F$, if $PM$ has $-1$ as an eigenvalue for every permutation matrix $P$, then the sum of
all entries of $M$ is $-k$.
\end{conj}

Next, using Corollary \ref{c4dd}, we prove that Conjecture \ref{guess} holds when $k=2$. For this, first we need a
lemma.
\begin{lem}\label{-1root2*2}
Let $M$ be a $2\times 2$ symmetric matrix over a field $F$ such that $-1$ is an eigenvalue of $PM$ for all permutation
matrices $P$. Then $$M=\begin{pmatrix} \alpha & -1-\alpha \\ -1-\alpha & \alpha \end{pmatrix}.$$
\end{lem}
\begin{proof}
From the two equations $\det(I+M)=0$ and $\det(I+P_{(1\ 2)}M)=0$ for the matrix $M=\begin{psmat} \alpha& \beta \\
\beta & \gamma \end{psmat}$, it follows that $\alpha\gamma=(\beta+1)^2$ and $\alpha+\gamma=-2(\beta+1)$. From this the
result follows easily.
\end{proof}
\begin{thm}\label{k=2}
If $G_{2\times n}$ is in standard form and is not \good, then the sum of the entries on any row of $G$ is $0$ and the
sum of the entries on any column of $G$ is $1$. In other words, Conjecture \ref{guess} is true in the case that $k=2$.
\end{thm}
\begin{proof}
We can assume that $n\geq 4$ and for simplicity we set $m=n-2$ and let $G=(I_2|A)$ where
$$A=\begin{pmatrix} a_1 & \cdots &a_{m} \\
 b_1& \cdots & b_m
\end{pmatrix}.$$
Note that $\Pi_2$ holds by Corollary \ref{c4dd}. Thus according to Proposition \ref{alpha}, there is a $\beta \in F$
such that $a_i+b_i=\beta$ for all $i$ and $m\beta^2=-2$. By Lemmas \ref{eigen} and \ref{-1root2*2}, the sum of each row
of $M=AA^t$ is 1. Therefore, we have $\sum_{i=1}^m  a_i^2 + \sum_{i=1}^m  a_ib_i=-1$. Replacing $b_i$ with $\beta-a_i$,
we deduce that $\beta\sum_{i=1}^m  a_i=-1$. Similarly, we have $\beta\sum_{i=1}^m  b_i=-1$. Consequently, if we let
$\alpha=\sum_{i=1}^m a_i$, then $\sum_{i=1}^m b_i=\alpha$ and $\alpha\beta=-1$. Therefore, we just need to prove
$\beta=1$.

First suppose that $\chr F=2$. Then from $\beta\neq 0$ and $m\beta^2=-2=0$ it follows that $m$ and hence $n$ are even.
Let $P=P_\sigma$ for $\sigma=(1\ 3)(2\ 4)(5 \ 6)(7\ 8) \cdots (n\ n-1)$. Then
$$GPG^t=\begin{pmatrix}
          0 & b_1+a_2+\gamma \\
          b_1 +a_2+\gamma & 0
        \end{pmatrix}\ \text{where} \ \gamma=\sum_{3\leq i\leq n \atop  i \text{ is odd } } a_ib_{i+1}+a_{i+1}b_i.$$
Using $b_i=\beta-a_i$, we get $\gamma=\beta\sum_{i=3}^m a_i=\beta(\alpha+a_1+a_2)=1+\beta a_1+\beta a_2$. Since
$\det(GPG^t)=0$, we see that $0=\beta+a_1+a_2+\gamma=(\beta+1)(a_1+a_2+1)$. If $\beta=1$, we are done. Else
$a_1+a_2=1$. If $n=4$, this means that $\alpha=1$ and again, we are done. Suppose $n>4$. By a similar argument
$a_i+a_j=1$ for each $1 \leq i\neq j\leq m$ and hence $a_i=a_j$ for all $i,j$. Then, since $m$ is even, we have
$\alpha=ma_1=0$, a contradiction, which concludes the proof in the case that $\chr F=2$.

Now assume that $\chr F\neq 2$. Let $\sigma=(1\ 3)(2 \ 4)$ and $\tau=(1\ 4)(2\ 3)$. Then the equations $\det(GP_\sigma
G^t)=0$ and $\det(GP_\tau G^t)=0$ read as
\begin{align}
    (2a_1+R)(2b_2+S)-(b_1+a_2+T)^2 & = 0 \label{eq1} \\
    (2a_2+R)(2b_1+S)-(b_2+a_1+T)^2 &= 0, \label{eq2}
\end{align}
where $R=\sum_{i>2}a_i^2$, $S=\sum_{i>2} b_i^2$ and $T=\sum_{i>2}a_ib_i$. Note that since $b_i=\beta-a_i$ the following
equations hold.
\begin{align}
b_2-b_1&=a_1-a_2 \label{eq3}\\
b_1^2-b_2^2+a_2^2-a_1^2& = 2\beta(a_2-a_1) \label{eq4} \\
a_1b_2-a_2b_1&=\beta(a_1-a_2) \label{eq5}
\end{align}
Also we have
\begin{equation} R+S+2T=\sum_{i>2} (a_i+b_i)^2=(m-2)\beta^2. \label{eq6}
\end{equation}
If we subtract \eqref{eq2} from \eqref{eq1} and simplify it using Equations \eqref{eq3}--\eqref{eq6}, we get
$$2(a_1-a_2)(4\beta+(m-2)\beta^2)=0.$$

Therefore, we have two cases. Case 1: $4\beta+(m-2)\beta^2\neq 0$. Then $a_1=a_2$ and by a similar argument $a_i=a_j$
and hence $b_i=b_j$ for all $i,j$. Thus $0\neq \alpha=ma_1=mb_1$ and we have $a_i=b_j=a$ for all $i,j$ and some $a\in
F$. Consequently, $2ma^2=\alpha\beta=-1$. If we set $x=(m-1)a^2=\f{-1}{2}-a^2$, then we have
$$0=\det(GP_{(1\ 3)}G^t)=2a-a^2+x=2a-2a^2-\f{1}{2}=-\f{(2a-1)^2}{2},$$
that is, $a=\f{1}{2}$ and $\beta=2a=1$, as required.

Case 2: $4\beta+(m-2)\beta^2 = 0$. Since $\beta\neq 0$, it holds that $4+(m-2)\beta=0$. As $\chr F>2$ and $4\neq 0$, it
follows that $\beta=\f{-4}{m-2}$. Now if we solve the equation $m\l(\f{-4}{m-2}\r)^2=m\beta^2=-2$, we get $m=-2$ and
hence $\beta=1$, which concludes the proof.
\end{proof}
Using Lemma \ref{gt}, we immediately get the following corollary.
\begin{cor}
In the case that $n=k+2$, Conjecture \ref{guess} holds true.
\end{cor}

\section{Conclusions and further research}
In this paper, we investigated \good\ codes, that is, linear $[n,k]$-codes such as $C$ for which there is another code
$D$, permutation equivalent to $C^\perp$, with $(C,D)$ being an LCP of codes. We observed that if the sum of all
entries on any row of a generating matrix $G$ of $C$ is zero and $(1,1,\ldots, 1)$ is in the row space of $G$, then $C$
is not
\good. We conjectured that the converse also holds and showed that, among several other special cases, in the cases that
$k\leq 2$ or $n-k\leq 2$ this conjecture holds. We also studied the sum of all entries of a $k\times k$ symmetric
matrix $M$ with the property that $-1$ is an eigenvalue of $PM$ for all permutation matrices and conjectured that this
sum is always $-k$. We showed that this holds for small $k$'s and for fields with large or zero characteristic and used
this to prove that several classes of codes are \good.

As further research, not only it remains to prove or reject Conjectures \ref{guess} and \ref{conj Pk}, but also it
remains to find efficient algorithms that for an \good\ code $C$, find its dual-equivalent code $D$ with $(C,D)$ an LCP
of codes. This can be used in cryptography against side channel and fault injection attacks. Also regarding Conjecture
\ref{conj Pk}, it should be mentioned that we used CoCoA computer software, to find all $k\times k$ matrices $M$ with
-1 as an eigenvalue of $PM$ for all permutation matrices $P$, when $k\leq 5$ over several small fields. All of them
satisfied a stronger condition than $\Pi_k$. In fact, either the sum of all entries on any row of them is $-1$ or the
sum of all entries on any column of them is $-1$. Note that if $M$ satisfies any of these two conditions, then $-1$ is
an eigenvalue of $PM$ for all permutation matrices $P$, since either $(I+PM)e^t=0$ or $e(I+MP)=0$ for $e=(1,\ldots,1)$.




\begin{thebibliography}{10}

\bibitem{cocoa} Abbott J., Bigatti A. M. and Robbiano L., CoCoA: a system for doing Computations in Commutative
    Algebra. Available at http://cocoa.dima.unige.it

\bibitem{Bhasin} Bhasin S., Danger J. L., Guilley S. Najm Z. and NgoX. T., Linear complementary dual code
    improvement to strengthen encoded circuit against hardware Trojan horses, \emph{IEEE International Symposium on Hardware
    Oriented Security and Trust (HOST)}, May 5177, 2015.

\bibitem{Borello}  Borello M., Cruz J. and Willems W., A note on linear complementary pairs of group codes,
    \emph{Disc.
    Math.}, \textbf{343}(8)(2020), 111905.

\bibitem{ODSM} Bringer J., Carlet C., Chabanne H., Guilley S. and  Maghrebi H., Orthogonal direct sum masking --- a
    smartcard friendly computation paradigm in a code, with builtin protection against side-channel and fault attacks,
    in \emph{WISTP}, Springer, Heraklion, 40-56, 2014.

\bibitem{carlet crypto} Carlet C. and Guilley S., Complementary dual codes for counter-measures to side-channel attacks,
    \emph{J. Adv. Math. Commun.}, \textbf{10}(1) (2016), 131--150.

\bibitem{CarletG} Carlet C., Güneri C., Özbudak F., Özkaya B. and Solé P., On linear complementary pairs of codes, \emph{IEEE
    Trans. Inform. Theory}, \textbf{64}(1)(2018), 6583–6588.

\bibitem{CarletM} Carlet C., Mesnager S., Tang C.M., Qi Y.F. and Pellikaam R., Linear codes over $F_q$ are equivalent to
    LCD codes for $q>3$, \emph{IEEE Trans. Inf. Theory},  \textbf{64}(4)(2018), 3010–3017.

\bibitem{GuneriM} Güneri C., Martínez-Moro E. and Sayıcı S., Linear complementary pair of group codes over finite chain
    rings, \emph{Des. Codes Cryptogr.},  \textbf{88}(2020), 2397-2405.

\bibitem{GuneriO} Güneri C., \"{O}zkaya B. and Sayıcı S., On linear complementary pair of $n$D cyclic codes, \emph{IEEE Commun.
    Lett.}, \textbf{22}(2018), 2404–2406.

\bibitem{Comb Fund Alg} Grinbeng D., Notes on the Combinatorial Fundamentals of Algebra, available
    online at: arXiv:2008.09862v3.

\bibitem{Hu} Hu P. and Liu X., Linear complementary pairs of codes over rings, \emph{Des. Codes Cryptogr.}, \textbf{89}(2021), 2495– 2509.

\bibitem{Jin} Jin L.F., Construction of MDS codes with complementary duals, \emph{IEEE Trans. Inf. Theory}, \textbf{63}(5)(2017), 2843–
    2847.

\bibitem{LiuX} Liu X. and Liu H., LCD codes over finite chain rings, \emph{Finite Fields Appl.} \textbf{34} (2015), 1--19.


\bibitem{LiuL} Liu H. and Liu X., LCP of matrix product codes, \emph{Lin. and Mult. Algebra},
    \textbf{70}(22)(2022), 7611–7622.

\bibitem{LiuW}
 Liu Z. and Wang J., Linear complementary dual codes over rings, \emph{Des. Codes Cryptogr.}, \textbf{87}(2019), 3077–3086.

\bibitem{Massey}
 Massey J. L., Linear codes with complementary duals, \emph{Disc. Math.}, \textbf{106/107}(1992), 337–342.

 \bibitem{roman} S. Roman, \emph{Coding and Information Theory}, Springer-verlag (1992).

\bibitem{Sendrier}
 Sendrier N., Linear codes with complementary duals meet the Gilbert-Varshamov bound, \emph{Disc. Math.},  \textbf{285}(2004), 345–347.


\end{thebibliography}
\end{document}